\documentclass[twocolumn,aps,pra,notitlepage,superscriptaddress,longbibliography]{revtex4-1}
\usepackage[breaklinks=true]{hyperref}
\usepackage{amsmath,amssymb,amsthm,mathtools,mathrsfs}
\usepackage[table,dvipsnames]{xcolor}
\usepackage{tikz}
\usepackage{adjustbox}
\usepackage{dsfont}
\usepackage[normalem]{ulem}
\usepackage{enumerate}
\usepackage[all]{xy} 
\usepackage{graphicx}
\usepackage{subfigure}

\theoremstyle{plain}
\newtheorem{theorem}{Theorem}
\newtheorem{lemma}{Lemma}

\theoremstyle{definition}
\newtheorem{definition}{Definition}
\newtheorem{remark}{Remark}
\newtheorem{example}{Example}

\newcommand{\bd}{\begin{definition}}
\newcommand{\ed}{\end{definition}}
\newcommand{\bt}{\begin{theorem}}
\newcommand{\et}{\end{theorem}}
\newcommand{\be}{\begin{equation}}
\newcommand{\ee}{\end{equation}}
\newcommand{\blem}{\begin{lemma}}
\newcommand{\elem}{\end{lemma}}
\newcommand{\bx}{\begin{example}}
\newcommand{\ex}{\end{example}}

\newcommand\define[1]{\emph{\textbf{#1}}}

\DeclareMathAlphabet{\mathpzc}{OT1}{pzc}{m}{it} 
 \DeclareFontFamily{OT1}{pzc}{}
 \DeclareFontShape{OT1}{pzc}{m}{it}{ <-> s*[1.2] pzcmi7t }{}
 \DeclareMathAlphabet{\mathpzc}{OT1}{pzc}{m}{it}
 
\def\R{{{\mathbb R}}}

\newcommand{\Tr}{\operatorname{Tr}}

\def\E{\mathcal{E}}
\def\N{\mathcal{N}}

\begin{document}									
\preprint{APS/123-QED}

\title{Quantum Mutual Information in Time}

\author{James Fullwood}
\email{fullwood@hainanu.edu.cn}
\affiliation{School of Mathematics and Statistics, Hainan University, Haikou, Hainan Province, 570228, China.}
\author{Zhen Wu}
\email{wzmath@sjtu.edu.cn}
\affiliation{School of Mathematical Sciences, MOE-LSC, Shanghai Jiao Tong University, Shanghai, 200240, China}
\author{Arthur J. Parzygnat}
\email{arthurjp@mit.edu}
\affiliation{Experimental Study Group, Massachusetts Institute of Technology, Cambridge, Massachusetts 02139, USA.}
\author{Vlatko Vedral}
\affiliation{Clarendon Laboratory, University of Oxford, Parks Road, Oxford OX1 3PU, United Kingdom.}

\date{\today}

\begin{abstract}
While the quantum mutual information is a fundamental measure of quantum information, it is only defined for spacelike-separated quantum systems. Such a limitation is not present in the theory of classical information, where the mutual information between two random variables is well-defined irrespective of whether or not the variables are separated in space or separated in time. Motivated by this disparity between the classical and quantum mutual information, we employ the pseudo-density matrix formalism to define a simple extension of quantum mutual information into the time domain. As in the spatial case, we show that such a notion of quantum mutual information in time serves as a natural measure of correlation between timelike-separated systems, while also highlighting ways in which quantum correlations distinguish between space and time. We also show how such quantum mutual information is time-symmetric with respect to quantum Bayesian inversion, and then we conclude by showing how mutual information in time yields a Holevo bound for the amount of classical information that may be extracted from sequential measurements on an ensemble of quantum states. 
\end{abstract}

	\maketitle

\section{Introduction}
The mutual information of two quantum systems $A$ and $B$ whose joint state is represented by a bipartite density matrix $\rho_{AB}$ is the real number $I(A:B)$ 
given by
\begin{equation}\label{MIX747}
I(A:B)=S(\rho_A)+S(\rho_B)-S(\rho_{AB})\, ,
\end{equation}
where $S(\cdot)$ denotes the von~Neumann entropy and $\rho_A$ and $\rho_B$ are the reduced density matrices of $\rho_{AB}$~\cite{CeAd97,NiCh11}. As $\rho_{AB}$ describes the joint state of two spacelike separated systems $A$ and $B$, it follows that the quantum mutual information is essentially a \emph{static} measure of information. This non-dynamical nature of quantum mutual information is in stark contrast with the notion of the classical mutual information $I(X:Y)$ of two random variables $X$ and $Y$, which is defined irrespective of whether $X$ and $Y$ are separated in space or in time~\cite{CoTh06}. 

In this work, we make use of the pseudo-density matrix (PDM) formalism \cite{FJV15,FuPa24} to define a temporal extension of quantum mutual information, and we investigate its properties. A PDM is a generalization of a density matrix that encodes correlations across both space and time, thus providing a single mathematical formalism for the study of both static and dynamical aspects of quantum information (see Refs.~\cite{MVVAPGDG21,ZPTGVF18,Pisar19,Marletto_2020,Marletto_2019,Liu_2024,LQDV23,FuPa24,song23,Jia_2023,PaFu24} for further development and applications of PDMs). In particular, PDMs provide a notion of a joint state $R_{AB}$ associated with timelike separated quantum systems $A$ and $B$ whose reduced density matrices are the individual states $\rho_A$ and $\rho_B$. However, while $R_{AB}$ is Hermitian and of unit trace, it is not positive in general. Thus, the extension of the density matrix formalism to include non-positive PDMs is akin to the way in which the metric of space is extended into the time domain in special relativity, where such an extension results in a metric of Lorentzian (as opposed to Euclidean) signature~\cite{French68,HaEl73}. 

As for defining quantum mutual information in time, note that since the logarithm is not uniquely defined for non-positive matrices, simply replacing $\rho_{AB}$ by $R_{AB}$ in equation \eqref{MIX747} will not in general yield a well-defined notion of quantum mutual information in time. To circumvent this issue, we use an extension of von~Neumann entropy to Hermitian matrices given by~\cite{FuPa23,PaTK23,Tu22,SSW14}
\begin{equation}\label{ENTX27}
S(X)=-\Tr\big[X\log|X|\big]\, .
\end{equation}
While there are various justifications for our choice of $S$ over other alternatives (such as those used e.g. in Refs.~\cite{Jia_2023,PaTK23}), a primary justification is that $S$ satisfies \cite{FuPa23} 

\begin{itemize}
\item
(Additivity) $S(X\otimes Y)=S(X)+S(Y)$ for all Hermitian matrices $X$ and $Y$.
\item
(Orthogonal Convexity) If $\{X^i\}$ is a collection of mutually orthogonal Hermitian matrices of unit trace and $p_i$ is a probability distribution, then 
\[
S\Big(\sum_{i}p_iX^i\Big)=H(p)+\sum_{i}p_iS(X^i)\, ,
\]
where $H(p)$ is the Shannon entropy of the probability distribution $p_i$. 
\end{itemize}
While the properties of additivity and orthogonal convexity are certainly desirable from a purely mathematical perspective, we show in Section~\ref{HQT} that additivity and orthogonal convexity of $S$ are also crucial for establishing a temporal analog of the Holevo bound~\cite{Holevo73}, thus providing an operational justification for the use of $S$.

Employing such a Hermitian extension $S$ of von~Neumann entropy and replacing $\rho_{AB}$ by a pseudo-density matrix $R_{AB}$ in equation \eqref{MIX747} yields a well-defined notion of quantum mutual information for timelike separated quantum systems. In what follows, we prove general statements regarding such a notion of mutual information in time and present many examples. As in the spatial case, we find that such a mutual information in time provides a natural measure of correlation between timelike separated systems, while also highlighting the way in which quantum correlations distinguish between space and time. In particular, for sequential measurements performed on a system of qubits at two times, we show in Section~\ref{MIT} that the mutual information vanishes for a system which is discarded and re-prepared between measurements, while it is maximized for a system undergoing unitary evolution between measurements. However, contrary to spatial mutual information, which for a pair of spacelike separated qubits can attain a maximum value of 2 in the presence of entanglement, we show in Section~\ref{QBT2X} that the mutual information in time between two timelike separated qubits never exceeds a value of 1. We argue that this is a reflection of the fact that, unlike spatial correlations, maximal correlations in time are non-monogamous~\cite{Marletto_2020}, as a qubit at a fixed point in time may be maximally correlated with a qubit both at a time in the future and a time in the past. 

While for mutual information in space it is 
manifest from definition~\eqref{MIX747} that $I(A:B)=I(B:A)$, for mutual information in time  the existence of such a symmetry is more subtle, as systems $A$ and $B$ may be related in time by a non-reversible process. However, we show in Section~\ref{TRS} that when the quantum channel establishing temporal correlations between $A$ and $B$ is \emph{Bayesian invertible} in the sense defined in Ref.~\cite{FuPa22a}, then it is fact that case that $I(A:B)=I(B:A)$, in accordance with the spatial case. 

\section{Pseudo-density matrices}

Let $\sigma_i$ denote the $i$th Pauli matrix for $i=0,...,3$, and let $m>0$ be a positive integer. For every $\alpha\in \{0,...,3\}^{m}$, we let $\alpha_{j}\in \{0,...,3\}$ denote the $j$th component of $\alpha$ for all $j\in \{1,...,m\}$, and then define $\sigma_{\alpha}$ as
\[
\sigma_{\alpha}=\sigma_{\alpha_1}\otimes \cdots \otimes \sigma_{\alpha_m}\, .
\]
Suppose a sequential measurement of $\sigma_{\alpha}$ followed by $\sigma_{\beta}$ is performed on a system of $m$-qubits at times $t_A$ and $t_B$ with $t_A<t_B$, and denote the Hilbert spaces of the system at the two times by $\mathcal{H}_A$ and $\mathcal{H}_B$, respectively. If the system evolves according to the quantum channel $\N:\mathcal{L}(\mathcal{H}_A)\to \mathcal{L}(\mathcal{H}_B)$ between measurements (we use $\mathcal{L}(\mathcal{H})$ to denote the algebra of linear operators on a Hilbert space $\mathcal{H}$), then the theoretical \emph{two-time expectation value} of the sequential measurement of $\sigma_{\alpha}$ followed by $\sigma_{\beta}$ is the real number $\langle  \sigma_{\alpha}\,, \sigma_{\beta} \rangle$ given by
\[
\langle  \sigma_{\alpha}\,, \sigma_{\beta} \rangle=\Tr\Big[\N(\Pi_{\alpha}^{+}\, \rho \Pi_{\alpha}^{+})\sigma_{\beta}\Big]-\Tr\Big[\N(\Pi_{\alpha}^{-}\, \rho \Pi_{\alpha}^{-})\sigma_{\beta}\Big]\, ,
\]
where $\rho$ is the initial state of the system, and $\Pi_{\alpha}^{\pm}=\frac{1}{2}(\mathds{1}\pm\sigma_{\alpha})$ is the orthogonal projection operator onto the $\pm 1$-eigenspace of $\sigma_{\alpha}$. 

In such a two-time measurement scenario, the associated \emph{pseudo-density matrix} (PDM) is defined to be the bipartite hermitian operator $R_{AB}$ on $\mathcal{H}_A\otimes \mathcal{H}_B$ given by
\begin{align} \label{PDMDFXS77}
R_{AB}= \frac{1}{4^{m}} \sum_{\alpha,\beta} \langle  \sigma_{\alpha}\,, \sigma_{\beta} \rangle \sigma_{\alpha} \otimes \sigma_{\beta}  \, .
\end{align}
It follows from the definition of PDM together with properties of Pauli matrices that for all $\alpha$ and $\beta$, 
\begin{equation}
\langle  \sigma_{\alpha}\,, \sigma_{\beta} \rangle=\Tr\Big[R_{AB}(\sigma_{\alpha}\otimes \sigma_{\beta})\Big]\, .
\end{equation}
Thus, PDMs extend the operational meaning of bipartite density matrices to the temporal domain for Pauli observables.

Although the definition of PDM is conceptually simple, it is not always practical for calculations. In Refs.~\cite{HHPBS17,LQDV23,FuPa24}, it was shown that the PDM $R_{AB}$ defined by \eqref{PDMDFXS77} may be given by the formula
\begin{equation} \label{SPTPDXS89}
R_{AB}=\frac{1}{2}\Big\{\rho\otimes \mathds{1}\, , \mathscr{J}[\N]\Big\}\, ,
\end{equation}
where $\{\cdot , \cdot \}$ denotes the anticommutator and $ \mathscr{J}[\N]=\sum_{i,j}|i\rangle \langle j |\otimes \N(|j\rangle \langle i |)$ is the \emph{Jamio{\l}kowski matrix} of the channel $\N$ \cite{Ja72}. 
When we wish to emphasize the dependence of the PDM $R_{AB}$ on the initial state $\rho$ and channel $\N$ as part of a two-time measurement scenario, we denote the RHS of Eq.~\eqref{SPTPDXS89} by $\N\star \rho$ and  refer to it as the \emph{spatiotemporal product} of the channel $\N$ and the initial state $\rho$.

The notion of PDM naturally extends to $n$-sequential measurements on a system of $m$-qubits for arbitrary $n>0$, yielding an operator $R_{A_1\cdots A_n}$ on the $n$-fold tensor product $\mathcal{H}_{A_1}\otimes \cdots \otimes \mathcal{H}_{A_n}$, where $A_i$ denotes the system at time $t_i$, with $t_1<\cdots <t_n$. We denote such an $n$-time PDM simply by $R_{1\cdots n}$. If the system evolves according to a channel $\N_i$ between measurements at times $t_i$ and $t_{i+1}$ for $i=1,...,n-1$, it was shown in Refs.~\cite{LQDV23,Fu23} that $R_{1\cdots n}$ may be given by the recursive formula
\begin{equation} 
\label{eq:multitimePDM}
R_{1\cdots n}=(\N_{n-1}\circ \Tr_{1\cdots (n-2)})\star R_{1\cdots (n-1)}\, ,
\end{equation}
where $\Tr_{1\cdots (n-2)}$ is the partial trace over the subsystems $A_1\cdots A_{n-2}$ and $\star$ denotes the spatiotemporal product.

We note that while PDMs are not positive in general, they are always Hermitian and of unit trace. Moreover, the reduced marginals onto a single factor $\mathcal{L}(\mathcal{H}_{A_i})$ are always density matrices, representing the state of the system at time $t_i$. Interestingly, positive PDMs admit dual interpretations as being extended across space on the one hand, and time on the other. We refer to such PDMs as \emph{dual states}. 

\section{Mutual Information in Time}\label{MIT}

Let $R_{1\cdots n}$ be an $n$-time PDM as in~\eqref{eq:multitimePDM}, let $A_i$ denote the associated system at time $t_i$ for $i=1,\dots,n$, let $k\in \{2,\ldots,n-1\}$, and let $R_{1\cdots k}$ and $R_{k+1\cdots n}$ be the PDMs corresponding to the first $k$ measurements and final $n-k$ measurements, respectively, which may be obtained from the PDM $R_{1\cdots n}$ by tracing out the associated complementary subsystems. We then define the \define{mutual information in time} between the joint temporal systems $A=A_{1}\cdots A_{k}$ and $B=A_{k+1}\cdots A_{n}$ to be the element $I(A:B)\in \R$ given by
\[
I(A:B)=S(R_{1\cdots k})+S(R_{k+1\cdots n})-S(R_{1\cdots n})\, ,
\]
where $S(\cdot)$ is the Hermitian extension of the von~Neumann entropy given by $S(X)=-\Tr\big[X\log|X|\big]$.

We now present some examples of mutual information in time.

\bx[A dual state]
Let $\rho_{AB}$ be the bipartite density matrix given by
\[
\rho_{AB}=
\left(
\begin{array}{cccc}
13/24&0&0&0\\
0&5/24&-1/6&0\\
0&-1/6&5/24&0\\
0&0&0&1/24\\
\end{array}
\right)\, .
\]
The density matrix $\rho_{AB}$ is an entangled state~\cite{We89} that has been shown to also be a 2-time PDM~\cite{song23}, and hence is a dual state. As such, the associated systems $A$ and $B$ may either be viewed as being spacelike separated or timelike separated. Thus, their mutual information in space is equal to their mutual information in time, which is approximately $0.2315$. 
\ex

\bx[A single qubit at two times]
Let $R_{AB}$ be the 2-time PDM associated associated a maximally mixed qubit that undergoes trivial dynamics between measurements, which is given by 
\[
R_{AB}=\frac{1}{2}{\tt SWAP}=\frac{1}{2}\left(
\begin{array}{cccc}
1&0&0&0\\
0&0&1&0\\
0&1&0&0\\
0&0&0&1\\
\end{array}
\right)\, .
\]
As the eigenvalues of $R_{AB}$ are $(1/2,1/2,1/2,-1/2)$, it follows that $I(A:B)=1$, which coincides with the von~Neumann entropy of the initial maximally mixed state. We note that although the temporal correlations between the qubit at two times in this example can violate CHSH inequalities \cite{BTCV04}, thus exhibiting maximal correlations in time, the mutual information between the qubit at two times is not 2, as one would expect from the case of a pair of maximally entangled qubits in space \cite{CeAd97}. We will further address this disparity between the temporal and spatial cases in Section~\ref{QBT2X}.  
\ex

The previous example is a special case of the following:

\bt \label{MTX1}
Suppose $R_{AB}=\mathcal{U}\star \rho$ is a 2-time PDM associated with a multi-qubit system initially in state $\rho$ that evolves according to a unitary channel $\mathcal{U}$ between measurements. Then $I(A:B)=S(\rho)$.
\et

\begin{proof}
Let $m$ be the number of qubits of the system represented by $R_{AB}$, let $n=2^m$, let $\mathfrak{mspec}(X)$ denote the multiset of eigenvalues of a matrix $X$, and suppose $\mathfrak{mspec}(\rho)=\{\lambda_1,\ldots,\lambda_n\}$. It then follows from Lemma~5.7 in Ref.~\cite{FuPa23} that
\[
\mathfrak{mspec}(\mathcal{U}\star \rho)=\mathfrak{mspec}(\rho)\cup \left.\Big\{\pm \frac{\lambda_i+\lambda_j}{2}\,\, \right| \,\, 0<i<j\leq n\Big\}\, .
\]
Since $f(x)=-x\log|x|$ is an odd function, it follows that
\begin{equation} \label{UNTXDXS347}
S(R_{AB})=S(\mathcal{U}\star \rho)=S(\rho)=S\big(\mathcal{U}(\rho)\big)\, ,
\end{equation}
where the final equality follows from the unitary invariance of the entropy function $S$. We then have
\[
I(A:B)=S(\rho)+S\big(\mathcal{U}(\rho)\big)-S(\mathcal{U}\star \rho)=S(\rho)\, ,
\]
as desired.
\end{proof}

The next two examples illustrate cases when the evolution between measurements is non-unitary.

\bx[Discard and prepare] \label{DSCXPRXP57}
Let $\mathcal{N}$ be the discard-and-prepare channel given by $\mathcal{N}(\rho)=\Tr[\rho]\sigma$ for some state $\sigma$. It then follows that if $R_{AB}$ is a 2-time PDM of the form $R_{AB}=\mathcal{N}\star \rho$, then $R_{AB}=\rho\otimes \sigma$ (and hence is a dual state). In such a case, the timelike separated systems $A$ and $B$ represented by the PDM $R_{AB}$ are such that $I(A:B)=0$. 
\ex

\bx[Decoherence]
Let $R_{AB}$ be the 2-time PDM associated with a single qubit in an initial state $\rho=|-\rangle \langle - |$, which between measurements is to evolve according to the decoherence map $\mathcal{D}$ given by
\[
\mathcal{D}\left(
\begin{array}{cc}
\rho_{00} & \rho_{01} \\
\rho_{10} & \rho_{11}
\end{array}
\right)=\left(
\begin{array}{cc}
\rho_{00} & 0 \\
0 & \rho_{11}
\end{array}
\right)\, .
\]
The associated PDM $R_{AB}=\mathcal{D}\star \rho$ has eigenvalues
\[
\left(\frac{1+\sqrt{2}}{4},\frac{1+\sqrt{2}}{4},\frac{1-\sqrt{2}}{4},\frac{1-\sqrt{2}}{4}\right)\, ,
\]
from which it follows that $I(A:B)\approx 0.79824$. In Ref.~\cite{brody2024} it is argued that contrary to conventional wisdom, decoherence is a process of information flowing \emph{into} the system from its environment (as opposed to the other way around). In support of this claim, we conjecture that in such a case $I(A:B)$ then quantifies the information gained by the system due to decoherence.  
\ex

In the next example, we consider a single qubit at multiple times, with a varying initial state.

\bx[A single qubit at multiple times]
\label{ex:sqmt}
Let $R_{1\cdots n}$ be a multiple time single qubit PDM with initial state $\rho=\text{diag}(p,1-p)$ with $p\in [0,1]$. 
Denote the channel responsible for the evolution of the system between times $t_i$ and $t_{i+1}$ by $\mathcal{N}_i$, and let $\E_{\eta}$ be the depolarizing channel given by
\[
\E_{\eta}(\rho)=\eta \rho +(1-\eta)\Tr[\rho]\frac{\mathds{1}}{2}\, ,
\]
where $\eta\in [0,1]$ is the depolarization parameter. We then consider several cases of mutual information in time where $\N_i$ is either the identity channel or the depolarizing channel $\E_{\eta}$ for $n=3$ and $n=4$, the results of which are plotted as functions of $p$ and $\eta$ in Figure~\ref{Fig6}.
\ex

\begin{figure*}
	\subfigure[]
 {\includegraphics[width=0.4\textwidth]{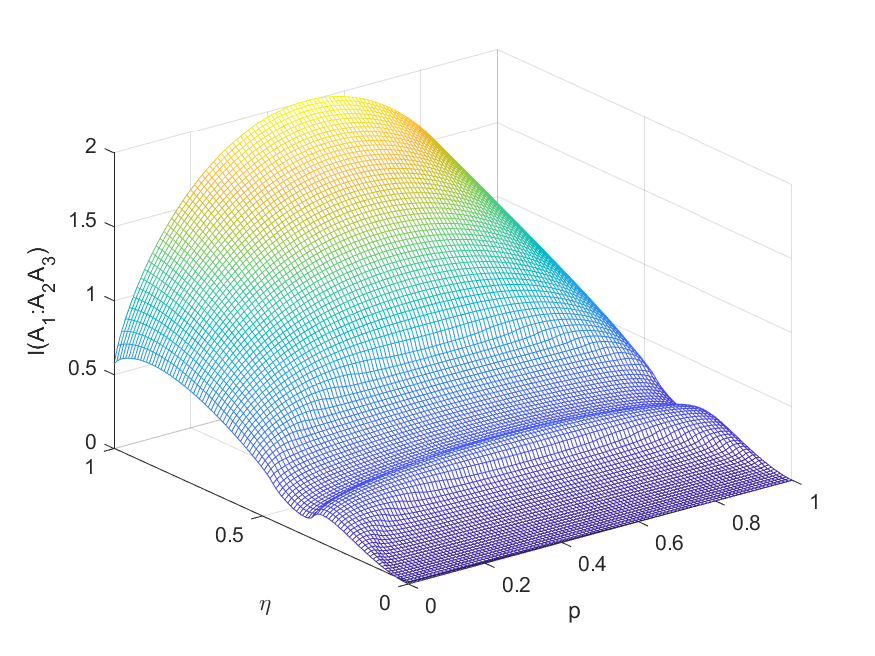}}
	\subfigure[]
 {\includegraphics[width=0.4\textwidth]{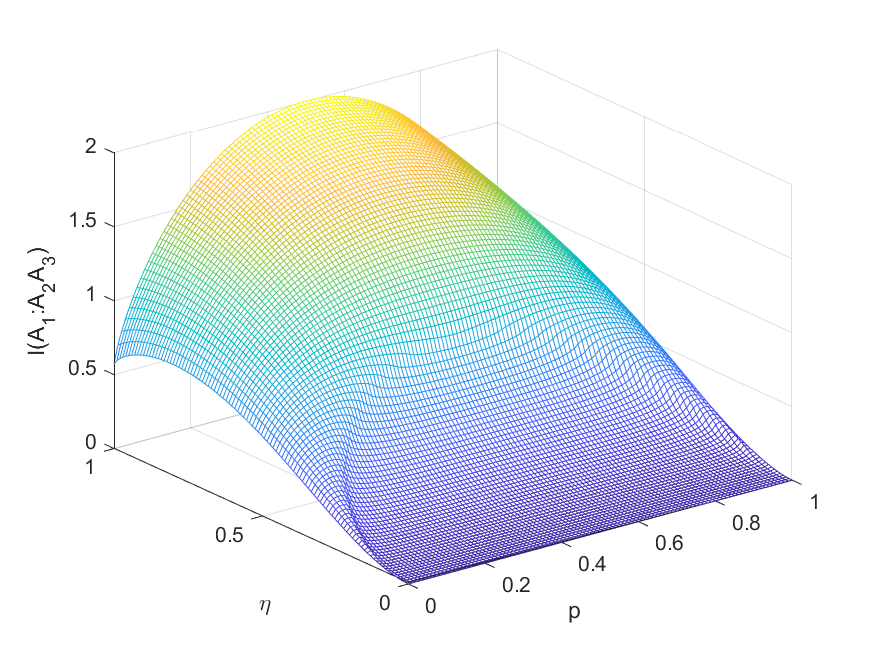}}
        \subfigure[]
    {\includegraphics[width=0.4\textwidth]{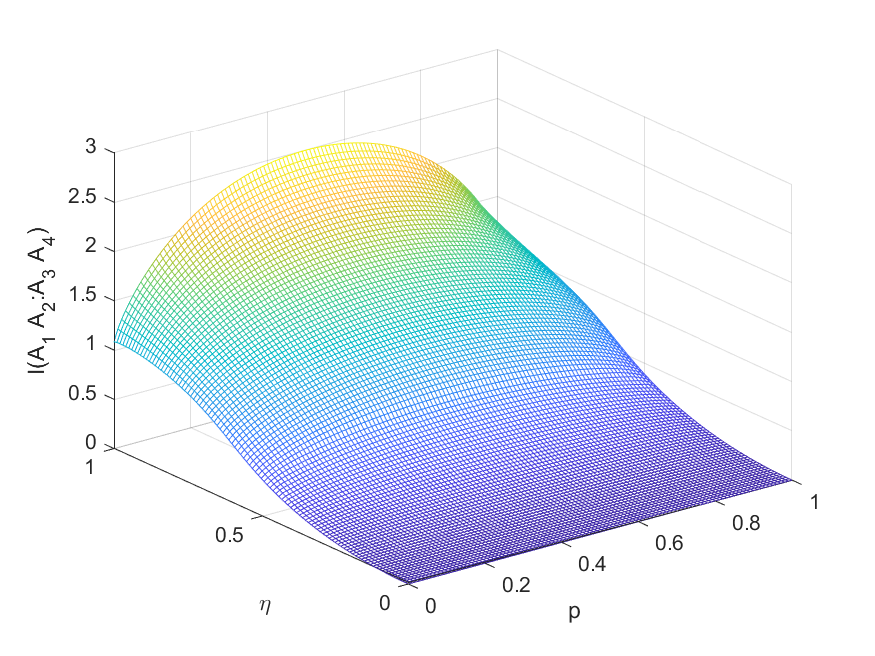}}
	\subfigure[]{\includegraphics[width=0.4\textwidth]{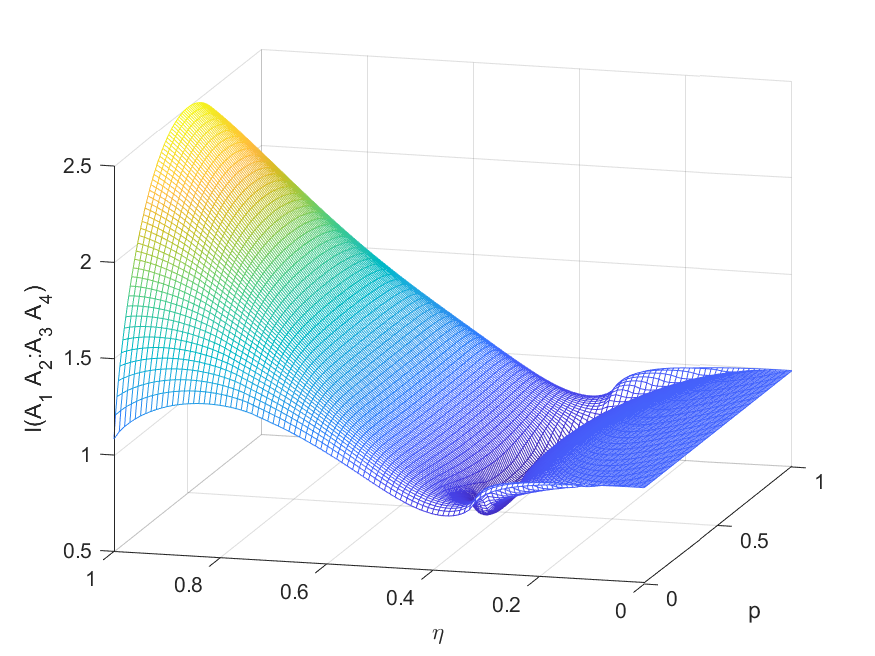}}
	\caption{{\label{Fig6}} \textbf{Mutual information between a single qubit at multiple times}. In each of the above graphs, $\eta$ is the depolarizing parameter of the channel $\E_{\eta}$ and $p$ is the parameter defining the initial state as in Example~\ref{ex:sqmt}. (a) $I(A_1:A_2A_3)$ for $n=3$ and $\N_1=\N_2=\E_{\eta}$.  (b) $I(A_1:A_2A_3)$ for $n=3$, $\N_1=\E_{\eta}$ and $\N_2=\text{id}$. (c) $I(A_1A_2:A_3A_4)$ for $n=4$, $\N_1=\N_3=\text{id}$ and $\N_2=\E_{\eta}$ (d) $I(A_1A_2:A_3A_4)$ for $n=4$, $\N_1=\N_3=\E_{\eta}$ and $\N_2=\text{id}$. In all cases, we find that the mutual information is maximized when $\eta=1$ and $p=0.5$, and minimized at 0 when $\eta=0$ for all $p$.}
\end{figure*}

\begin{remark}[Lack of monotonicity]
While the spatial quantum mutual information satisfies monotonicity, namely, $I(A:\Gamma(B))\leq I(A:B)$ for every local operation $\Gamma$ on the system $B$, local operations can either increase or decrease mutual information in time. In particular, while tracing out a subsystem $A_i$ of a multipartite spatial system $A_1\cdots A_n$ essentially discards the information contained in $A_i$, tracing out a subsystem $A_i$ of a multipartite \emph{dynamical} system $A_1\cdots A_n$ results in a redistribution of information across time to form a distinct dynamical system $A_1\cdots A_{i-1}A_{i+1}\cdots A_n$, where $A_{i-1}$ then has direct causal influence on $A_{i+1}$. As such, the local operation of tracing out a subsystem of a dynamical system $A_1\cdots A_n$ may result in an \emph{increase} in the mutual information between different timelike separated subsystems, thus resulting in a failure of monotonicity.
\end{remark}

\section{A single qubit at two times: The general case}\label{QBT2X}

For a single qubit at two-times we obtain the following theorem, whose proof can be found in Appendix~\ref{app:Ibounds}.

\bt \label{TMX387}
Let $R_{AB}$ be a 2-time PDM associated with a single qubit whose initial state is maximally mixed and which evolves according to a channel that is either unital or has Choi rank no greater than 2 between measurements. Then $0\leq I(A:B)\leq 1$.
\et

While we conjecture that $0\leq I(A:B)\leq 1$ holds for an arbitrary qubit at two times, our proof of Theorem~\ref{TMX387} relies on the hypotheses that the initial state is maximally mixed and that the channel corresponding to the evolution between measurements is either unital or has Choi rank no greater than 2. We note that it follows from Theorem~\ref{MTX1} that the bound $0\leq I(A:B)\leq 1$ also holds for arbitrary initial states provided that the evolution between measurements is unitary. 

Assuming $0\leq I(A:B)\leq 1$ holds for an arbitrary qubit at two times, such a feature of temporal correlations is in stark contrast with the case of two qubits separated in space, which can achieve a maximum mutual information of 2 in the presence of entanglement \cite{CeAd97}. This is a reflection of the fact that unlike the monogamy spatial correlations, a single qubit can be maximally correlated both with a qubit in its immediate future and it immediate past. In particular, given a 3-time measurement scenario, the qubit at time $t_2$ can violate CHSH inequalities both with the qubits at times $t_1$ and $t_3$, which is a consequence of the fact that two-time correlation functions associated with Pauli observables are independent of initial conditions \cite{BTCV04,Fritz10}.  

An interesting consequence of the bound $I(A:B)\leq 1$ for a qubit at two times, is that if a pair of spatially entangled qubits has a mutual information greater than 1, then it follows that the two qubits do not admit a dual description of being timelike separated. The converse to this statement however does not hold in general. In particular, a pair of spatially entangled qubits not admitting a dual description of being timelike separated can have mutual information less than one, as in the case of the entangled qubits described by the pure state $\sqrt{p}|00\rangle +\sqrt{1-p}|11\rangle$ for small $p$. 

We also note that it follows from Example~\ref{DSCXPRXP57} that the lower bound of $0$ for the mutual information $I(A:B)$ in the context of Theorem~\ref{TMX387} is achieved when a qubit evolves according to a completely depolarizing channel between measurements. This is consistent with the fact that all information about the initial measurement is lost in such a scenario. On the other hand, we know from Theorem~\ref{MTX1} that the upper bound of 1 for $I(A:B)$ is achieved when a maximally mixed qubit undergoes unitary evolution between the two measurements, which is consistent with the fact that no information is lost (or gained) under unitary evolution.

\section{Time-reversal Symmetry}\label{TRS}

For spatial mutual information it is always the case that $I(A:B)=I(B:A)$, as there is no preferred directionality in space between spacelike separated systems. In particular, applying the swap transformation $\mathscr{S}:\mathcal{L}(\mathcal{H}_A)\otimes \mathcal{L}(\mathcal{H}_B)\to \mathcal{L}(\mathcal{H}_B)\otimes \mathcal{L}(\mathcal{H}_A)$ to a bipartite density matrix $\rho_{AB}$ yields a density matrix $\rho_{BA}$ which is physically indistinguishable from $\rho_{AB}$, so that $I(A:B)=I(B:A)$. 

When $A$ and $B$ are timelike separated the situation is more subtle, as a quantum channel $\mathcal{N}$ establishing temporal correlations between $A$ and $B$ may be non-reversible due to noise. In such a case, the non-reversibility of $\mathcal{N}$ is reflected in the associated PDM $R_{AB}=\mathcal{N}\star \rho$ by the fact that the swap $\mathscr{S}(R_{AB})$ is not a PDM. However, when $\mathscr{S}(R_{AB})$ \emph{is} in fact a PDM, so that $\mathscr{S}(R_{AB})=R_{BA}=\mathcal{M}\star \sigma$, where $\sigma=\mathcal{N}(\rho)$ and $\mathcal{M}$ is a quantum channel from $B$ to $A$, then $\mathcal{M}$ is said to be a \emph{Bayesian inverse} of the channel $\mathcal{N}$ \emph{with respect to} the state $\rho$ \cite{FuPa22a}. As such, the notion of a Bayesian inverse provides a novel form of reversibility for timelike separated quantum systems, even when the systems are temporally correlated via a noisy channel. 

As solving for a Bayesian inverse $\mathcal{M}$ involves solving a linear system, the number of solutions is either 0,1 or $\infty$, with 0 corresponding to no Bayesian inverse existing, 1 corresponding to a unique Bayesian inverse and $\infty$ corresponding to the notion of Bayesian inverse being non-unique. However, in the case of non-uniqueness, given two Bayesian inverses $\mathcal{M}$ and $\mathcal{M}'$ of $\mathcal{N}$ with respect to $\rho$, then it follows from the definition of Bayesian inverse that $\mathcal{M}\star \sigma=\mathcal{M}'\star \sigma$, so that the notion of a PDM $R_{BA}$ corresponding to a time-reversal of the process described by $R_{AB}$ is well-defined.

In the context of mutual information in time, suppose we have a PDM $R_{AB}=\mathcal{N}\star \rho$ with $\mathcal{N}$ Bayesian invertible with respect to $\rho$, so that the swap $\mathscr{S}(R_{AB})=R_{BA}=\mathcal{M}\star \sigma$, where $\sigma=\mathcal{N}(\rho)$ and $\mathcal{M}$ is a quantum channel from $B$ to $A$ corresponding to a Bayesian inverse of $\mathcal{N}$ with respect to $\rho$. Now since the eigenvalues of a matrix are invariant under the swap map $\mathscr{S}$, it follows that $S(R_{AB})=S(R_{BA})$. Moreover, since $\sigma=\mathcal{N}(\rho)$ and $\rho=\mathcal{M}(\sigma)$, we have
\begin{align*}
I(B:A)&=S(\sigma)+S(\mathcal{M}(\sigma))-S(R_{BA}) \\
&=S(\rho)+S(\mathcal{N}(\rho))-S(R_{AB}) \\
&=I(A:B)\, .
\end{align*}
Thus, the existence of a Bayesian inverse of a channel establishes a symmetry in time which extends the  symmetry $I(A:B)=I(B:A)$ of spatial mutual information into the time domain.

\section{A Holevo Bound in Time} \label{HQT}

Suppose Alice sends a pure state $\rho^i$ to Bob with probability $p_i$, after which Bob performs a measurement on $\rho^i$ corresponding to a POVM $\{M_j\}$. Then the probability $q_{j|i}$ of the measurement outcome $M_j$ given the state $\rho^i$ is given by $q_{j|i}=\Tr[\rho^i M_j]$. If $X$ is a random variable associated with the probability distribution $p_i$, and $Y$ is a random variable associated with the probability distribution $q_j=\sum_{i}p_iq_{j|i}$, then the classical mutual information $I(X:Y)$ is bounded above by the \emph{Holevo quantity} $\chi$ given by 
\begin{equation}\label{HQXT457}
\chi=S\Big(\sum_{i}p_i\rho^i\Big)-\sum_{i}p_iS(\rho^i)\, .
\end{equation}
The inequality $I(X:Y)\leq \chi$ is referred to as the \emph{Holevo bound}~\cite{Holevo73}, and thus provides a fundamental limit on the amount of classical information which may be extracted from an ensemble of quantum states. Moreover, if we form the joint classical-quantum state $\rho_{AB}$ associated with Alice's classical register and Bob's random acquisition of the state $\rho^i$, namely, 
\begin{equation}\label{CQTXPDM67}
\rho_{AB}=\sum_{i}p_i |i\rangle \langle i | \otimes \rho^i\, ,
\end{equation}
then it is straightforward to show that the mutual information $I(A:B)$ associated with $\rho_{AB}$ is such that $I(A:B)=\chi$. This, together with the Holevo bound, provides an operational interpretation for the quantum mutual information $I(A:B)$. In this section, we obtain analogous results for mutual information in time.

For this, we first note that from a dynamical perspective, the state $\rho_{AB}$ is in fact the PDM given by $\rho_{AB}=\mathcal{P}\star \rho_{A}$, where $\mathcal{P}$ is the channel given by $\mathcal{P}(|i\rangle \langle j|)=\delta_{ij}\rho^i$ and $\rho_{A}$ is the state given by $\rho_{A}=\sum_{i}p_i|i\rangle \langle i|$. In such a 2-time scenario, the channel $\mathcal{P}$ is then viewed as a channel from Alice to Bob. If Bob then makes two sequential measurements on $\rho^i$, with the system evolving according to a channel $\mathcal{N}$ between measurements, then it follows from \eqref{eq:multitimePDM} that the associated 3-time PDM (with $\mathcal{N}_1=\mathcal{P}$ and $\mathcal{N}_2=\mathcal{N}$) is then given by 
\begin{equation}\label{CLXQXS57}
R_{AB_1B_2}=\sum_{i}p_i |i\rangle \langle i|\otimes R_{12}^{i}\, ,
\end{equation}
where $R_{12}^{i}=\mathcal{N}\star \rho^i$ is the 2-time PDM associated with Bob's sequential measurements. We then obtain the following:

\begin{theorem} \label{ABX2TX79}
Let $R_{12}^{i}$ be the 2-time PDM associated with Bob's sequential measurements of $\rho^i$, and let $I(A:B_1B_2)$ be the mutual information associated with the classical-quantum PDM $R_{AB_1B_2}$ given by \eqref{CLXQXS57}.
Then the following statements hold.
\begin{enumerate}[i.]
\item \label{ABX2TX791}
$I(A:B_1B_2)=S\Big(\sum_{i}p_iR^i_{12}\Big)-\sum_{i}p_iS(R^i_{12})$.
\item \label{ABX2TX792}
If Bob's system evolves unitarily between measurements, then $I(A:B_1B_2)=\chi$, where $\chi$ is the Holevo quantity \eqref{HQXT457}.
\end{enumerate}
\end{theorem}

\begin{proof}
First, we have $\Tr_{B_1B_2}\big[R_{AB_1B_2}\big]=\sum_{i}p_i |i\rangle \langle i|$ and
$\Tr_{A}\big[R_{AB_1B_2}\big]=\sum_{i}p_i R_{12}^{i}$,
so that
\begin{equation} \label{MXI57}
I(A:B_1B_2)=H(p)+S\Big(\sum_{i}p_i R_{12}^{i}\Big)-S(R_{AB_1B_2})\, ,
\end{equation}
where $H(p)$ denotes the Shannon entropy associated with the probability distribution $p_i$. Now since $S$ is additive and orthogonally convex (as defined in the introduction), we have
\begin{align*}
S(R_{AB_1B_2})&=S\Big(\sum_{i}p_i |i\rangle \langle i|\otimes R_{12}^{i}\Big) \\
&=H(p)+\sum_{i}p_iS\Big(|i\rangle \langle i|\otimes R_{12}^{i}\Big) \\
&=H(p)+\sum_{i}p_i\left(S(|i\rangle \langle i|)+S(R_{12}^{i})\right) \\
&=H(p)+\sum_{i}p_iS(R_{12}^{i})\, . 
\end{align*} 
Substituting $S(R_{AB_1B_2})=H(p)+\sum_{i}p_iS(R_{12}^{i})$ into equation \eqref{MXI57} then yields
\begin{equation}\label{EXS71}
I(A:B_1B_2)=S\Big(\sum_{i}p_iR_{12}^{i}\Big)-\sum_{i}p_iS(R_{12}^{i})\, ,
\end{equation}
which proves item \ref{ABX2TX791}. Now if Bob's system evolves according to a unitary channel $\mathcal{U}$ between measurements, so that $R^i_{12}=\mathcal{U}\star \rho^i$, it follows from \eqref{UNTXDXS347} that $S(R^i_{12})=S(\rho^i)$. Moreover, as the spatiotemporal product $\star$ is linear in $\rho$, we have
\begin{align*}
S\Big(\sum_{i}p_iR_{12}^{i}\Big)&=S\Big(\sum_{i}p_i(\mathcal{U}\star \rho^{i})\Big)=S\Big(\mathcal{U}\star \sum_{i}p_i \rho^i\Big) \\
&=S\Big(\sum_{i}p_i \rho^i\Big)\, ,
\end{align*}
where the final equation follows again from \eqref{UNTXDXS347}. It then follows from equation \eqref{EXS71} that $I(A:B_1B_2)=\chi$, which proves item \ref{ABX2TX792}.
\end{proof}

Now let $X$ be the random variable associated with Alice's classical register, and let $(Y_1,Y_2)$ be the bivariate random variable corresponding to the joint statistics of Bob's classical register when performing sequential measurements of the POVMs $\{M_j\}$ and $\{N_k\}$. We assume that after a measurement outcome of $M_j$ when measuring $\rho^i$ that Bob's system updates to  
$U\sqrt{M_j}\rho^i\sqrt{M_j}U^{\dag}$ (suitably normalized) for some fixed unitary $U$ which is independent of $i$ and $j$. It then follows that the joint distribution $p_{jk}$ corresponding to the bivariate random variable $(Y_1,Y_2)$ is given by $p_{jk}=\sum_{i}p_{jk|i}$, where the conditional distribution $p_{jk|i}$ is given by
\[
p_{jk|i}=\Tr\Big[U\sqrt{M_j}\hspace{0.25mm}\rho^i\sqrt{M_j}U^{\dag}N_k\Big]\, .
\]
Now since the mapping $\mathcal{N}:\mathcal{L}(\mathcal{H}_A)\to \mathcal{L}(\mathcal{H}_{Z_1})\otimes \mathcal{L}(\mathcal{H}_{Z_2})$ given by 
\[\mathcal{N}(\rho)=\sum_{j,k}\Tr\Big[U\sqrt{M_j}\hspace{0.25mm}\rho \sqrt{M_j}U^{\dag}N_k\Big] |j\rangle \langle j|\otimes |k\rangle \langle k|
\]
is a quantum channel, it follows from Lemma~\ref{LXMHXS57} of Appendix~\ref{HLSX} that the mutual information $I(X:Y_1Y_2)$ between Alice's classical register $X$ and the bivariate register $(Y_1,Y_2)$ corresponding to Bob's sequential measurements is also bounded above by the Holevo quantity $\chi$. Thus, it follows from Theorem~\ref{ABX2TX79} that
\[
I(X:Y_1Y_2)\leq I(A:B_1B_2)\, ,
\]
where again $I(A:B_1B_2)$ is the mutual information in time associated with the classical-quantum PDM $R_{AB_1B_2}$ given by \eqref{CLXQXS57}. By taking a supremum over all possible POVMs $\{M_j\}$ and $\{N_k\}$, it then follows that the amount of classical information that may be extracted from such sequential measurements is bounded above by the mutual information in time $I(A:B_1B_2)$.

\section{Concluding Remarks}

In this work, we have employed the pseudo-density matrix formalism to define an extension of quantum mutual information into the time domain. Such a measure of quantum information naturally quantifies correlations between timelike separated quantum systems, and recovers the usual spatial mutual information for systems which admit a dual description of either being timelike or spacelike separated. We have proved such a notion of mutual information in time satisfies various properties one would expect from such a dynamical measure of information, highlighting the ways in which quantum correlations distinguish between space and time. 

We note that while an alternative formulation of quantum mutual information for timelike separated systems has recently appeared in \cite{Glorioso_24}, such a quantity is defined in terms of optimizing a relative entropy over all possible system-ancilla coupling schemes, resulting in a quantity which often diverges, even for systems consisting of a single qubit. From the dynamical perspective espoused in this work, the limitations of the aforementioned approach stem from the fact that a spatial notion of relative entropy is being utilized, rather than a dynamical extension of relative entropy using the entropy function $S(X)=-\Tr[X\log|X|]$.

As for the quantum mutual information in time defined in this work, we have shown that a single qubit at two times can never exceed a mutual information of $1$. This is consistent with the fact that for two qubits to be correlated in time there must necessarily be a causal transfer of information between the two qubits via a quantum channel, and hence two timelike separated qubits can share at most one bit of information. An interesting consequence of such an upper bound for the mutual information between a qubit at two times, is that if two spatially separated qubits have mutual information exceeding 1, then it necessarily follows that the correlations between such qubits may not be realized in a dynamical setting where the two qubits are separated in time. 

Another interesting aspect of quantum mutual information in time is that in all known 2-time examples it is non-negative, despite the fact that the entropy function $S(X)=-\Tr[X\log|X|]$ is not subadditive on general hermitian matrices of unit trace. In particular, in Ref.~\cite{FuPa23}, an example of a unit trace bipartite Hermitian matrix $\rho_{AB}\in \mathcal{L}(\mathcal{H}_A\otimes \mathcal{H}_B)$ was constructed whose marginals are density matrices, and yet 
\be \label{FSXA57}
S(\rho_A)+S(\rho_B)-S(\rho_{AB})<0\, .
\ee
In such a case, $\rho_{AB}$ is of the form $\rho_{AB}=\mathcal{N}\star \rho_A$ with $\mathcal{N}$ \emph{not} a CPTP map, which the entropy function $S$ seems to be detecting via the failure of subadditivity \eqref{FSXA57}. 

Finally, it is worth noting that quantum mutual information in time may be used to define a notion of channel capacity for a quantum channel $\mathcal{N}$ which is directly analogous to Shannon's original definition in the classical case, i.e., by taking a supremum of mutual information with respect to all possible input states for the channel $\mathcal{N}$. It would be interesting then to compare such a notion of channel capacity with other notions of quantum capacity \cite{Lloyd_1997,Devetak_2005,Holevo_2020}, especially since such a quantity is often computable.  

\section{Acknowledgements}

VV acknowledges support from the Templeton and the Gordon and Betty Moore foundations, and WZ would like to thank the Institute of Mathematical Science at Hainan University for hospitality during the course of this project.

\bibliography{references}

\clearpage
\newpage

\title{Methods}
\author{testing}

\maketitle
\onecolumngrid
\vspace{1cm}

\begin{center}\large \textbf{Quantum Mutual Information in Time} \\
\textbf{--- Supplementary Material ---}\\
\end{center}

\appendix

\section{Proof of Theorem~\ref{TMX387}}
\label{app:Ibounds}

Let $A$ and $B$ denote quantum systems corresponding to a single qubit at times $t_A$ and $t_B$ with $t_A<t_B$, and suppose the system which is initially in the state $\rho$ evolves according to a quantum channel $\mathcal{N}:\mathcal{L}(\mathcal{H}_A)\to \mathcal{L}(\mathcal{H}_B)$ between measurements at times $t_A$ and $t_B$. In such a case, the mutual information $I(A:B)$ between the qubit at times $t_A$ and $t_B$ is given by
\begin{equation}\label{MXS89}
I(A:B)=S(\rho)+S(\mathcal{N}(\rho))-S(R_{AB})\, ,
\end{equation}
where 
\begin{equation}\label{PDXMS17}
R_{AB}=\mathcal{N}\star \rho\equiv \frac{1}{2}\Big\{\rho\otimes \mathds{1}\,, \mathscr{J}[\mathcal{N}]\Big\} \quad , \quad \mathscr{J}[\N]=\sum_{i,j}|i\rangle \langle j |\otimes \N(|j\rangle \langle i |)\, ,
\end{equation}
and
\[
S(X)=-\Tr\Big[X\log|X|\Big]
\]
for all Hermitian matrices $X$\, . In this section, we will prove Theorem~\ref{TMX387}, which states that if $\rho=\mathds{1}/2$ is maximally mixed, and if $\mathcal{N}$ is either unital or of Choi rank no more than 2, then $0\leq I(A:B)\leq 1$. We note that under the assumption that $\rho=\mathds{1}/2$ is maximally mixed, it follows from \eqref{PDXMS17} that
\begin{equation}\label{PDMXM67}
R_{AB}=\frac{1}{2}\mathscr{J}[\N]\, 
\end{equation}
for every channel $\mathcal{N}:\mathcal{L}(\mathcal{H}_A)\to \mathcal{L}(\mathcal{H}_B)$.

\subsection{Reduction to \texorpdfstring{$\mathcal{N}_d$}{}}   

Let $\mathcal{N}:\mathcal{L}(\mathcal{H}_A)\to \mathcal{L}(\mathcal{H}_B)$ be a quantum channel. As an arbitrary density matrix $\rho\in \mathcal{L}(\mathcal{H}_A)$ may be written with respect to the Pauli basis $\{\sigma_i\}_{i=0}^{3}$ as
\[
\rho = \frac{1}{2}(I + \Vec{r} \cdot \vec{\sigma})
\]
with $\vec{r} \in \mathbb{R}^3$ and $|\vec{r}|\leq 1$, it can be shown that the quantum channel $\mathcal{N}$ takes the form
\[
\mathcal{N}(\rho) =  \frac{1}{2}\bigg[I + \bigg(\textbf{t}+\textbf{T}(\Vec{r})\bigg) \cdot \vec{\sigma}\bigg]\, .
\]
where $\textbf{t} \in \mathbb{R}^3$ and $\textbf{T} \in \mathbb{R}^{3\times 3}$. Thus, a matrix representation $N$ of the channel $\mathcal{N}$ with respect to the Pauli basis is of the form \cite{RSW02}
\begin{equation} \label{MTXRPX37}
    N = 
    \begin{pmatrix}
        1 & \textbf{0} \\
        \textbf{t} & \textbf{T}
    \end{pmatrix}\, ,
\end{equation}
Moreover, one can show that there exist unitary channels $\mathcal{U}$ and $\mathcal{V}$ such that
 \begin{equation}\label{unitaryinvariant}
    \mathcal{N}= \mathcal{V}\circ \mathcal{N}_{d} \circ \mathcal{U}\, ,
\end{equation}
where the matrix representation of $\mathcal{N}_{d}$ is given by 
\begin{equation}\label{diag}
    N_d = \left(
    \begin{array}{cccc}
        1 & 0 & 0 & 0 \\
        t_1 & \lambda_1 & 0 & 0 \\
        t_2 & 0 & \lambda_2 & 0 \\
        t_3 & 0 & 0 & \lambda_3 \\
    \end{array}
    \right)\, .
\end{equation}

Now let $\rho=\mathds{1}/2$ be the maximally mixed state, and let $R_{AB}$ and $R^{(d)}_{AB}$ be the PDMs given by $R_{AB}=\mathcal{N}\star \rho$ and $R^{(d)}_{AB}=\mathcal{N}_d\star \rho$. Since $\mathcal{N}= \mathcal{V}\circ \mathcal{N}_{d} \circ \mathcal{U}$, the covariance property of PDMs \cite{PFBC23,Fu23a} yields
\begin{equation}\label{CVX87}
(\mathcal{U}^{*}\otimes \mathcal{V})(\mathcal{N}_d\star \rho)=\mathcal{N}\star \mathcal{U}^*(\rho)\, ,
\end{equation}
where $\mathcal{U^*}=\mathcal{U}^{-1}$ is the Hilbert-Schmidt adjoint of $\mathcal{U}$. If $U\in \mathcal{L}(\mathcal{H}_A)$ is the unitary operator such that $\mathcal{U}(a)=UaU^{\dag}$ for all $a\in \mathcal{L}(\mathcal{H}_A)$, then it follows that
\[
\mathcal{U}^*(\rho)=U^{\dag}\rho U=U^{\dag} (\mathds{1}/2)U=(\mathds{1}/2)U^{\dag}U=\rho\, ,
\]
thus by \eqref{CVX87} we have
\[
(\mathcal{U}^{*}\otimes \mathcal{V})(R^{(d)}_{AB})=(\mathcal{U}^{*}\otimes \mathcal{V})(\mathcal{N}_d\star \rho)=\mathcal{N}\star \rho=R_{AB}\, .
\]
Moreover, since $\mathcal{U}(\rho)=\rho$ we have $\mathcal{N}(\rho)=\mathcal{V}(\mathcal{N}_d(\mathcal{U}(\rho)))=\mathcal{V}(\mathcal{N}_d(\rho))$. It then follows by the unitary-invariance of $S$ that $S(R_{AB})=S(R^{(d)}_{AB})$ and $S(\mathcal{N}(\rho))=\mathcal{S}(\mathcal{N}_d(\rho))$. Now let $I(A:B)$ and $I(A:B)^{(d)}$ denote the mutual information associated with $R_{AB}$ and $R_{AB}^{(d)}$ respectively. We then have
\[
I(A:B)=S(\rho)+S(\mathcal{N}(\rho))-S(R_{AB})=S(\rho)+S(\mathcal{N}_d(\rho))-S(R_{AB}^{(d)})=I(A:B)^{(d)}\, ,
\] 
thus $I(A:B)=I(A:B)^{(d)}$ for any qubit channel $\mathcal{N}$ with initial state $\rho=\mathds{1}/2$. As such, when proving Theorem~\ref{TMX387}, we may compute $I(A:B)$ in terms of $I(A:B)^{(d)}$, which we'll refer to as ` reduction to $\mathcal{N}_d$'. 

\subsection{Unital channels}

In this section, we consider the case when $\mathcal{N}$ is a unital channel. Before proceeding with the proof, we recall that a real vector $\vec{y}=(y_1,\cdots,y_n)$ is said to \emph{majorize} a real vector $\vec{x}=(x_1,\cdots,x_n)$, written $\vec{x} \prec \vec{y}$ if and only if $\sum_{i=1}^k x_{[i]} \leq \sum_{i=1}^k y_{[i]}$ for $k=1,\cdots, n-1$, where $x_{[1]}\geq \cdots \geq x_{[n]}$ denotes the components of $\vec{x}$ in decreasing order and $y_{[1]}\geq \cdots\geq y_{[n]}$ denotes the components of $\vec{y}$ in decreasing order. A function $f$ is then said to be \emph{Schur-Concave} if and only if $x \prec y\implies f(\vec{x}) \geq f(\vec{y})$ \cite{MOA10}.

So now suppose $\mathcal{N}$ is a unital channel, in which case the vector $\textbf{t}$ in \eqref{MTXRPX37} is zero. This then implies that $\mathcal{N}_d$ is a Pauli channel, i.e., $\mathcal{N}_d(\rho) = \sum_{i=0}^3 p_i \sigma_i \rho \sigma_i$ for some probability vector $(p_0,p_1,p_2,p_3)$ and for all $\rho\in \mathcal{L}(\mathcal{H}_A)$. Since $\mathcal{N}_d$ is unital, it follows that $\mathcal{N}_d(\rho)=\rho$ for $\rho=\mathds{1}/2$ maximally mixed, so that the associated mutual information $I(A:B)$ is then given by
\begin{equation}\label{MXSTX71}
I(A:B)=2-S(R_{AB})\, .
\end{equation}
As for the associated PDM $R_{AB}=\mathcal{N}_d\star \rho$, we have
\begin{equation}
    R_{AB}\overset{\eqref{PDMXM67}}= \frac{1}{2} \mathscr{J}[\mathcal{N}] = \frac{1}{2}\left(
    \begin{array}{cccc}
        p_0 + p_3 & 0 & 0 &  p_1-p_2 \\
        0 & p_1+p_2 & p_0-p_3 & 0 \\ 
        0 & p_0-p_3& p_1+p_2 & 0 \\
        p_1-p_2 & 0 & 0 & p_0+p_3 \\
    \end{array}
    \right)\, ,
\end{equation}
from which it follows that the eigenvalues of $R_{AB}$ are 
\[
\left(\frac{p_1+p_2 \pm |p_0-p_3|}{2}\,, \frac{p_0+p_3 \pm |p_1-p_2|}{2}\right) = \left(\frac{1}{2} - p_0\,, \frac{1}{2} - p_1\,,\frac{1}{2} - p_1\,,\frac{1}{2} - p_3\right)\, .
\]

We now consider two cases:
\begin{itemize}
    \item[(1)] If all $p_i\leq 1/2$ for all $i$, then the eigenvalues of $R_{AB}$ are all non-negative and $R_{AB}$ is a dual state, thus $S(R_{AB}) \leq \log(4)= 2$. We now show $S(R_{AB}) \geq 1$. For this, let $\vec{\lambda}=(\lambda_1,\ldots,\lambda_4)$ denote the probability vector of eigenvalues of $R_{AB}$ with $\lambda_1\geq \cdots \geq \lambda_4$. Since $p_i\leq 1/2$ for $i=0,...,3$ it follows that $\lambda_j\leq 1/2$ for $j=1,...,4$, thus the probability vector $\vec{y} = (1/2,1/2,0,0)$ majorizes $\vec{\lambda}$, hence $1 = H(\vec{y}) \leq H(\vec{\lambda})=S(R_{AB})$ since the Shannon entropy is Schur-Concave. It then follows that $0\leq I(A:B)\leq 1$ by \eqref{MXSTX71}. 
    \item[(2)] If there exists a $p_i > 1/2$, with loss of generality, assume that $p_0 >1/2$ and $p_i<1/2$ for $i=1,2,3$. We then have 
    \begin{equation} \label{ENTXSR87}
\begin{aligned}
    S(R_{AB}) & = -\left(\frac{1}{2}-p_0\right)\log\left(p_0-\frac{1}{2}\right)- \sum_{i=1}^3 \left(\frac{1}{2}-p_i\right)\log\left(\frac{1}{2}-p_i\right) \\
    & = \left(p_0-\frac{1}{2}\right)\log\left(p_0-\frac{1}{2}\right) - \left(\frac{1}{2}+p_0\right)\log\left(\frac{1}{2}+p_0\right) - \left(\frac{1}{2}+p_0\right)\sum_{i=1}^3 \frac{\frac{1}{2}-p_i}{\frac{1}{2}+p_0} \log\bigg( \frac{\frac{1}{2}-p_i}{\frac{1}{2}+p_0}\bigg)\\
    & =: S(R_{AB})_{p_0}\, ,
\end{aligned}
\end{equation}
where we use the notion $S(R_{AB})_{p_0}$ when we want to emphasize that we may view $S(R_{AB})$ as a function of $p_0$. We now wish to bound $S(R_{AB})$ from above and below. To obtain a minimum value of $S(R_{AB})$, we only need compute the minimum value of $S(R_{AB})$ for every fixed $p_0$, i.e., 
\[
\min S(R_{AB}) = \min_{p_0} \min_{p_1,p_2,p_3} S(R_{AB})_{p_0}.
\]
Now for a fixed $p_0$, it follows from \eqref{ENTXSR87} that $S(R_{AB})_{p_0}$ is minimized by finding $p_1,p_2,p_3$ which minimize the entropy of the vector $\vec{v} = (\frac{1/2-p_1}{p_0+1/2}, \frac{1/2-p_2}{p_0+1/2}, \frac{1/2-p_3}{p_0+1/2})$. As $p_i < 1/2$ for $i=1,2,3$, it follows that the vector $\vec{u} = (\frac{p_0-1/2}{p_0+1/2}, \frac{1/2}{p_0+1/2}, \frac{1/2}{p_0+1/2}),$ i.e. $p_1 = 1-p_0$ and $p_2 = p_3 = 0$, majorizes $\vec{v}$, hence 
\[
-\frac{p_0-1/2}{p_0+1/2} \log \left(\frac{p_0-1/2}{p_0+1/2}\right) - \frac{1}{p_0+1/2}\log \left(\frac{1/2}{p_0+1/2}\right) = H(\vec{u}) \leq H(\vec{v})\, , 
\]
so that 
\begin{equation}
    S(R_{AB})_{p_0} \geq \left(p_0-\frac{1}{2}\right)\log\left(p_0-\frac{1}{2}\right) - \left(p_0-\frac{1}{2}\right)\log\left(p_0-\frac{1}{2}\right) + 1 = 1\, .
\end{equation}
Therefore, for every fixed $p_0$, we always have $S(R_{AB})_{p_0} \geq 1$, thus $S(R_{AB}) \geq 1$. To bound $S(R_{AB})$ from above, notice that the vector $\vec{v}$ majorizes the vector $\vec{w} = (1/3, 1/3, 1/3)$ and $\vec{w}$ can always be attained by choosing $p_i = \frac{1-p_0}{3}, i=1,2,3$ for every fixed $p_0 > 1/2$, hence
\[
S(R_{AB})_{p_0} \leq  \left(p_0-\frac{1}{2}\right)\log\left(p_0-\frac{1}{2}\right) - \left(\frac{1}{2}+p_0\right)\log\left(\frac{1}{2}+p_0\right) + \left(\frac{1}{2}+p_0\right) \log(3)\, .
\]
Moreover, since
\[
\frac{d}{dp_0} S(R_{AB})_{p_0} =  \log\left(p_0-\frac{1}{2}\right)-\log\left(p_0+\frac{1}{2}\right) + \log(3) \leq 0\, ,
\]
it follows that $S(R_{AB}) \leq S(R_{AB})_{1/2} \leq \log(3)$, thus $2-\log(3)\leq I(A:B)\leq 1$ by \eqref{MXSTX71}. 
\end{itemize}
By reduction to $\mathcal{N}_d$, it follows that if $\mathcal{N}$ is unital, then $0\leq I(A:B)\leq 1$, as desired. We note that in such a case we have $I(A:B) = 0$ if and only if $\mathcal{N}_d$ is the completely depolarizing channel, and $I(A:B) = 1$ if and only if there are at least two $p_i$ in $\mathcal{N}_d$ equal to zero, i.e., when the Choi rank of $\mathcal{N}_d$ is no more than 2. 

\subsection{Channels of Choi rank no greater than 2}

In this section, we consider the case when the Choi rank of $\mathcal{N}$ is no greater than 2. In such a case, one may show that for the matrix representation $N_d$ of the channel $\mathcal{N}_d$ as given by \eqref{diag}, we have that $t_1=t_2=0$, $\lambda_3 = \lambda_1\lambda_2$ and $t_3^2 = (1-\lambda_1^2)(1-\lambda_2^2)$. Therefore, let $\lambda_1 = \cos{u}$ and $\lambda_2 = \cos{v}$, so that the matrix $N_d$ may be written in the form
\begin{equation}
    N_d = 
    \begin{pmatrix}
        1 & 0 & 0 & 0 \\
        0 & \cos{u} & 0 & 0 \\
        0 & 0 & \cos{v} & 0 \\
        \sin{u}\sin{v} & 0 & 0 & \cos{u}\cos{v} \\
    \end{pmatrix},
\end{equation}
with $u\in [0,2\pi), v\in [0,\pi)$. It is then straightforward to verify that with respect to this parameterization, $\mathcal{N}_d$ admits the Kraus representation given by \cite{RSW02} 
\[
\mathcal{N}_d(\rho)=K_+ \rho K_+^{\dag}+K_- \rho K_-^{\dag}\, ,
\]
where
\begin{equation}\label{Kraus}
\begin{aligned}
    K_+ &= \bigg[\cos{\frac{1}{2}v}\cos{\frac{1}{2}u}\bigg] I + \bigg[\sin{\frac{1}{2}v}\sin{\frac{1}{2}u}\bigg] \sigma_z = 
    \begin{pmatrix}
        \cos{\frac{v-u}{2}} & 0 \\
        0 & \cos{\frac{v+u}{2}}
    \end{pmatrix},\\
    K_- &= \bigg[\sin{\frac{1}{2}v}\cos{\frac{1}{2}u}\bigg] \sigma_x -i \bigg[\cos{\frac{1}{2}v}\sin{\frac{1}{2}u}\bigg] \sigma_y = 
    \begin{pmatrix}
        0 & \sin{\frac{v-u}{2}} \\
        \sin{\frac{v+u}{2}} & 0 
    \end{pmatrix}.
\end{aligned}
\end{equation}

Since $\mathcal{N}_d$ is non-unital, we have $\sin{u}\sin{v}>0$, so that for $\rho=\mathds{1}/2$ the PDM $R_{AB}=\mathcal{N}_d\star \rho$ is then given by
\begin{equation}
    R_{AB} \overset{\eqref{PDMXM67}}=\frac{1}{2} \mathscr{J}[\mathcal{N}] = \left(
    \begin{array}{cccc}
        \cos^2{\frac{v-u}{2}} & 0 & 0 & \sin{\frac{v-u}{2}}\sin{\frac{v+u}{2}} \\
        0 & \sin^2{\frac{v+u}{2}} & \cos{\frac{v-u}{2}}\cos{\frac{v+u}{2}} & 0 \\
        0 & \cos{\frac{v-u}{2}}\cos{\frac{v+u}{2}} & \sin^2{\frac{v-u}{2}} & 0 \\
        \sin{\frac{v-u}{2}}\sin{\frac{v+u}{2}} & 0 & 0 & \cos^2{\frac{v+u}{2}}
    \end{array}
    \right)\, ,
\end{equation}
hence the eigenvalues of $R_{AB}$ are 
\[
\left(\frac{1}{2},\frac{1}{2},\frac{\cos{u}\cos{v}}{2},-\frac{\cos{u}\cos{v}}{2}\right)\, .
\]
It then follows that $S(R_{AB})=S(\rho)=1$, which yields
\[
I(A:B)=S(\rho)+S(\mathcal{N}_d(\rho))-S(R_{AB})=S(\mathcal{N}_d(\rho))\, .
\]
Since $\mathcal{N}_d(\rho)$ is a density matrix $0\leq S(\mathcal{N}_d(\rho))\leq 1$, thus $0\leq I(A:B)\leq 1$. By reduction to $\mathcal{N}_d$, it follows that if the Choi rank of $\mathcal{N}$ is no greater than 2, then $0\leq I(A:B)\leq 1$, as desired. 

\section{Holevo bound for sequential measurements}\label{HLSX}

In this section, we prove a lemma which implies that the Holevo bound also holds for \emph{sequential} measurements on an ensemble of quantum states. 

\begin{lemma}\label{LXMHXS57}
Let $X$ be a random variable whose values are associated with a probability distribution $p_i$, let $p_{jk|i}$ be a family of bivariate probability distributions conditioned on $i$, and suppose there exists a quantum channel $\mathcal{N}:\mathcal{L}(\mathcal{H}_A)\to \mathcal{L}(\mathcal{H}_{Z_1})\otimes \mathcal{L}(\mathcal{H}_{Z_2})$ and a collection of density matrices $\{\rho^i\}\subset \mathcal{L}(\mathcal{H}_A)$ such that $\mathcal{N}(\rho^i)=\sum_{j,k}p_{jk|i} |j\rangle \langle j|\otimes |k\rangle \langle k|$ for all $i$.Then for every bivariate random variable $(Y_1,Y_2)$ whose values are associated with the bivariate distribution $p_{jk}=\sum_{i}p_ip_{jk|i}$, 
\[
I(X:Y_1Y_2)\leq S\Big(\sum_{i}p_i\rho^i\Big)-\sum_{i}p_iS(\rho^i)\, ,
\]
where $I(X:Y_1Y_2)$ is the mutual information of the random variables $X$ and $(Y_1,Y_2)$.
\end{lemma}

\begin{proof}
By the Stinespring dilation theorem there exists an isometry $\mathcal{V}:\mathcal{L}(\mathcal{H}_A)\to \mathcal{L}(\mathcal{H}_{Z_1})\otimes \mathcal{L}(\mathcal{H}_{Z_2})\otimes \mathcal{L}(\mathcal{H}_{Z_3})$ and an operator $V$ such that $\mathcal{V}(\rho)=V\rho V^{\dag}$ for all $\rho\in \mathcal{L}(\mathcal{H}_{A})$ and $\mathcal{N}(\rho)=\Tr_{Z_3}\Big[V\rho V^{\dag}\Big]$ for all $\rho\in \mathcal{L}(\mathcal{H}_A)$. Now let $\rho_{XZ_1Z_2Z_3}$ be the density matrix given by
\[
\rho_{XZ_1Z_2Z_3}=\sum_{i}p_i|i\rangle \langle i|\otimes V\rho^i\hspace{0.25mm}V^{\dag}\, ,
\]
and let $\rho_{XZ_1Z_2}$ be the density matrix given by
\[
 \rho_{XZ_1Z_2}=\sum_{i,j,k}p_{ijk}|i\rangle \langle i|\otimes |j\rangle \langle j|\otimes |k\rangle \langle k|\, ,
 \]
where $p_{ijk}=p_ip_{jk|i}$. We then have
\begin{align*}
\Tr_{Z_3}\Big[\rho_{XZ_1Z_2Z_3}\Big]&=\Tr_{Z_3}\Big[\sum_{i}p_i|i\rangle \langle i|\otimes V\rho^xV^{\dag}\Big]=\sum_{i}p_i|i\rangle \langle i|\otimes \Tr_{Z_3}\Big[V\rho^iV^{\dag}\Big]  \\
&=\sum_{i}p_i|i\rangle \langle i|\otimes \mathcal{N}(\rho^i)=\sum_{i}p_i|i\rangle \langle i|\otimes \Big(\sum_{j,k}p_{jk|i} |j\rangle \langle j|\otimes |k\rangle \langle k|\Big)  \\
&=\sum_{i,j,k}p_ip_{jk|i}|i\rangle \langle i|\otimes |j\rangle \langle j|\otimes |k\rangle \langle k|=\sum_{i,j,k}p_{ijk}|i\rangle \langle i|\otimes |j\rangle \langle j|\otimes |k\rangle \langle k| \\
&=\rho_{XZ_1Z_2}\, ,
\end{align*}
thus monotonicity of (spatial) quantum mutual information implies $S(X:Z_1Z_2)\leq S(X:Z_1Z_2Z_3)$, where $S(X:Z_1Z_2)$ is the quantum mutual information associated with the density matrix $\rho_{XZ_1Z_2}$ and $S(X:Z_1Z_2Z_3)$ is the quantum mutual information associated with the density matrix $\rho_{XZ_1Z_2Z_3}$.  Moreover, we have
\begin{align*}
S(X:Z_1Z_2Z_3)&=S\Big(\sum_{i}p_i|i\rangle \langle i|\Big)+S(\rho_{Z_1Z_2Z_3})-S(\rho_{XZ_1Z_2Z_3}) \\
&=H(p)+S\Big(\sum_{i}p_iV\rho^iV^{\dag}\Big)-S\Big(\sum_{i}p_i|i\rangle \langle i|\otimes V\rho^i\hspace{0.25mm}V^{\dag}\Big) \\
&=H(p)+S\Big(\sum_{i}p_i\rho^i\Big)-\Big(H(p)+\sum_{i}p_iS(\rho^i)\Big) \\
&=S\Big(\sum_{i}p_i\rho^i\Big)-\sum_{i}p_iS(\rho^i)\, ,
\end{align*}
and since $I(X:Y_1Y_2)=S(X:Z_1Z_2)$, we then have
\[
I(X:Y_1Y_2)=S(X:Z_1Z_2)\leq S(X:Z_1Z_2Z_3)=S\Big(\sum_{i}p_i\rho^i\Big)-\sum_{i}p_iS(\rho^i)\, ,
\]
as desired.
\end{proof}

\end{document}